\documentclass[12pt]{iopart}

\usepackage{amsthm}
\usepackage{nicefrac}
\usepackage{iopams}  
\usepackage{graphics}

\newcommand{\ZZ}{\mathbb Z}

\newcommand{\CC}{\mathbb C}        
\newcommand{\TT}{\mathbb T}       
\newcommand{\OO}{\mathbb O}			
\newcommand{\YY}{\mathbb Y}
\newcommand{\DD}{\mathbb D}

\newcommand{\mf}[1]{\mathfrak{#1}}
\newcommand{\GL}{\textrm{GL}}
\newcommand{\bigslant}[2]{{\raisebox{.2em}{$#1$}\left/\raisebox{-.2em}{$#2$}\right.}}

\newcommand\zn[1]{\bigslant{\ZZ}{#1}}  

\newtheorem{thm}{Theorem}[section]
\newtheorem{lem}[thm]{Lemma}
\newtheorem{cor}[thm]{Corollary}

\newtheorem{example}[thm]{Example}
\newtheorem{definition}[thm]{Definition}
\newtheorem{prop}[thm]{Proposition}

\newenvironment{mymatrix}{\left(\!\!\begin{array}{*{20}{c}}}{\end{array}\!\!\right)}

\begin{document}

\title{Automorphic Lie Algebras with dihedral symmetry}

\author{V Knibbeler$^1$, S Lombardo$^1$ and J A Sanders$^2$}

\address{1 Department of Mathematics and Information Sciences, Northumbria University, UK}
\address{2 Department of Mathematics, Vrije Universiteit Amsterdam, The Netherlands}
\eads{\mailto{vknibbeler@gmail.com}
}

\begin{abstract}
The concept of \emph{Automorphic Lie Algebras} 
arises in the context of reduction groups introduced in the early 1980s in the field of integrable systems. 
Automorphic Lie Algebras are obtained by imposing a discrete group symmetry on a current algebra of Krichever-Novikov type. 
Past work shows remarkable uniformity between  algebras 
associated to different reduction groups. 
For example, if the base Lie algebra is $\mathfrak{sl}_2(\mathbb{C})$ and the poles of the Automorphic Lie Algebra are restricted to an exceptional orbit of the symmetry group, changing the reduction group does not affect the Lie algebra structure.
In the present research we fix the reduction group to be the dihedral group and vary the orbit of poles as well as the group action on the base Lie algebra. We find a uniform description of Automorphic Lie Algebras with dihedral symmetry, valid for poles at exceptional and generic orbits.

\end{abstract}

\vspace{2pc}
\maketitle

\section{Introduction}
Automorphic Lie Algebras were introduced in the context of the classification of integrable partial differential equations \cite{Lombardo,LM04iop,LM05comm}. In fact, the Zakharov-Shabat / Ablowitz-Kaup-Newell-Segur scheme, 
used to integrate these equations, requires a pair of matrix valued functions of the spectral parameter $\lambda \in\overline{\CC}$ living on the Riemann sphere. Since a general pair of such $\lambda$-dependent matrices gives rise to an under determined system of differential equations, one requires additional constraints. By a well established scheme introduced by Mikhailov \cite{Mikhailov81}, and further developed in \cite{LM05comm}, this can be achieved by imposing a group symmetry on the matrices. Algebras consisting of all such symmetric matrices are called \emph{Automorphic Lie Algebras}, by analogy with automorphic functions. Since their introduction they have been extensively studied (see, \cite{LS10} and references therein, but also \cite{Bury2010,Chopp2011}). 
\begin{definition}[Automorphic Lie Algebras,  \cite{Lombardo,LM05comm}]
Let $\mathfrak{g}$ be a simple Lie algebra, $\mathcal{M}(\overline{\mathbb{C}})$ the field of meromorphic functions on the Riemann sphere, and $G$ a finite subgroup of $\textrm{Aut}(\mathfrak{g}\otimes\mathcal{M}(\overline{\mathbb{C}}))$; the \emph{Automorphic Lie Algebra} is the space of invariants 
\[\left(\mathfrak{g}\otimes\mathcal{M}(\overline{\mathbb{C}})\right)^G_\Gamma=
\left\{f\in\mf{g}\otimes\mathcal{M}(\overline{\mathbb{C}})\;|\;f\textrm{ analytic on }\overline{\CC}\setminus\Gamma,\; gf=f \;\forall g\in G\right\}\] where $\Gamma\subset\overline{\mathbb{C}}$ is a single $G$-orbit where poles are allowed.
\end{definition}
Algebras of meromorphic functions, taking values in a finite dimensional Lie algebra and having their poles restricted to a finite set, are \emph{current algebras of Krichever-Novikov type} \cite{schlichenmaier2003higher,schlichenmaier2013virasoro}. 
Automorphic Lie Algebras are subalgebras thereof.  
They are infinite dimensional over the complex numbers but finitely generated over a ring.
The infinite dimensional view dominates in the current literature. In this paper we rather consider the Lie algebras as modules over a polynomial ring (in the invariants, or the modular invariant); in fact, while this adds some computational complications, one is rewarded with classical looking Serre normal form results (see Section \ref{sec:nf}).

The case where the reduction group $G$ is the symmetry group of a polygon, the dihedral group $\DD_N$, has been studied before. Seminal but isolated results were obtained already in \cite{Mikhailov81}; the first attempt towards a systematic classification can be found in \cite{LM04iop}, further development in \cite{LM05comm}, and \cite{LS10}. However, all results so far have been obtained under the simplifying assumption that one can use the same dihedral representation to define an action on either the space of vectors or matrices, and on the polynomials. Here this restriction is lifted.
In the context of integrable systems, dihedral symmetric Automorphic Lie Algebras have been studied the most 
(e.g.  \cite{LM04iop,Mikhailov81,Wa09,konstantinou2013darboux,MPW2014}) and together with icosahedral symmetry, dihedral symmetry seems the most relevant in physical systems. 

In this paper the complete classification of Automorphic Lie Algebras with dihedral symmetry is tackled using classical invariant theory. The Riemann sphere is identified with the complex projective line $\CC P^1$ consisting of quotients $\frac{X}{Y}$ of two complex variables by setting $\lambda=\frac{X}{Y}$. M\"obius transformations on $\lambda$ then correspond to linear transformations on the vector $(X,Y)$ by the same matrix. Classical invariant theory is used to find the $G$-invariant subspaces of $\CC[X,Y]$-modules, where $\CC[X,Y]$ is the ring of polynomials in $X$ and $Y$. These ring-modules of invariants are then localised by a choice of multiplicative set of invariants. This choice corresponds to selecting the orbit $\Gamma$ of poles. The set of elements in the localisation of degree zero, i.e. the set of elements which can be expressed as functions of $\lambda$, generate the Automorphic Lie Algebra.
Once the algebra is identified, it is transformed into a normal form in the spirit of the standard Serre normal form \cite{Humphreys}; we believe this is the most convenient form for analysis.

This investigation leads to our main result:

%

\emph{
Let the dihedral group $\DD_N$ act faithfully on the spectral parameter $\lambda\in\overline{\CC}$ and let $\Gamma\subset\overline{\CC}$ be an orbit thereof.
Let $V$ be an irreducible complex projective representation of $\DD_N$ and $\mf{g}(V)$ a simple Lie subalgebra and $\DD_N$-submodule of $\textrm{End}(V)\cong V\otimes V^\ast$. This leaves only $\mf{sl}(V)$. Then the Automorphic Lie Algebra $(\mf{g}(V)\otimes\mathcal{M}(\overline{\mathbb{C}}))^{\DD_N}_\Gamma$ 
is generated by 
three matrices as a module over the ring of automorphic functions and there is a choice of generators 
$e_0,\;e_+,\;e_-$ in which the bracket relations are 
\[[e_0,e_{\pm}]=\pm 2 e_\pm\qquad[e_+,e_-]=fe_0\] where $f$ is the automorphic function with zeros of lowest possible order at all exceptional orbits not equal to $\Gamma$. To be more specific, $f=\frac{F_0^{\nu_0}F_1^{\nu_1}F_2^{\nu_2}}{F_\Gamma^{3\nu_{\Gamma}}}$ where $\Gamma_i$ denotes an exceptional orbit and $F_i$ respectively $F_\Gamma$ is the form of lowest degree that vanishes on $\Gamma_i$, respectively $\Gamma$, as introduced in Section \ref{sec:pre}. If $\Gamma$ is one of the exceptional orbits the Lie algebra structure can be summarised as in Table \ref{tab:structure constants}.
\begin{table}[h!]
\caption{Lie algebra structure of Automorphic Lie Algebras with dihedral symmetry:
the structure constant $f$ with poles at the exceptional orbits $\Gamma_i$.
}
\begin{center}
\begin{tabular}{c|cccc}
\label{tab:structure constants}
 $\Gamma$ & $\Gamma_0$ & $\Gamma_1$ & $\Gamma_2$ \\
\hline
$f$ & $\frac{F_1^2F_2^N}{F_0^4}$ & $\frac{F_2^NF_0^2}{F_1^4}$ & $\frac{F_0^2F_1^2}{F_2^{2N}}$\\
\end{tabular}
\end{center}
\end{table}
}

In particular, if one evaluates the algebra elements at $\lambda$ then $(\mathfrak{g}(V)\otimes\mathcal{M}(\overline{\mathbb{C}}))^{\DD_N}_\Gamma(\lambda)\cong \mf{sl}_2(\CC)$ if and only if the orbit $\DD_N\lambda$ is generic and not equal to $\Gamma$.
The result highlighted above (see Theorem \ref {thm:main2} in Section \ref{sec:nf}) allows us to provide a complete classification of $\mathfrak{sl}_2(\CC)$-based Automorphic Lie Algebras with dihedral symmetry, which is an important step towards the complete classification of Automorphic Lie Algebras, as the simplifying assumption made in \cite{LS10} to use the same matrices representing the reduction group to act on the spectral parameter as well as on the base Lie algebra is dropped. This simplifying assumption is no longer valid when considering higher dimensional Lie algebras.
Moreover, there is no need to distinguish between \emph{generic} and \emph{exceptional} orbits, as we establish that they can be all treated in the same way, in contrast to what was previously generally accepted (\cite{LM05comm,Bury2010}), even if the corresponding algebras are non isomorphic.

This paper is organised as follows: Section \ref{sec:pre} summarises definitions and basic notions used in the rest of the paper. Here notation is also fixed. Section \ref{sec:tALIAS} illustrates our approach to the classification of Automorphic Lie Algebras, while Section \ref{sec:D_NALIAS} presents the results for the $\DD_N$ case. In particular, in \ref{sec:nf} the Serre normal forms are presented and in \ref{sec:moi} the concept of \emph{matrices of invariants} is introduced. We believe that the latter will play a fundamental role in the classification of higher dimensional Automorphic Lie Algebras.
A section of Conclusions and an Appendix complete the paper.

\section{Preliminaries}
\label{sec:pre}
In this section, for the benefit of the reader, we recall definitions and basic notions from the theory of representations of finite groups, notably $\DD_N$ representation theory, and from the classical theory of invariants which will be used in what follows. For representation theory of finite groups we refer the reader to \cite{Serre}, among many other books. The results from invariant theory are discussed in \cite{stanley1979invariants}.

\subsection{Representations of the dihedral group $\DD_N$}
The dihedral group of order $2N$ can be defined abstractly by
\[\DD_N=\langle r, s\;|\;r^N=s^2=(rs)^2=1 \rangle.\]
If one thinks of the representation as symmetries of a regular $N$-gon, then $r$ stands for rotation and $s$ for reflection. The group consists of $N$ rotations $r^\ell$ and $N$ reflections $sr^\ell$, $0\le \ell <N$.
There are significant differences between the cases where $N$ is odd or even. For instance, the square of a rotation of order $N$ has again order $N$ when $N$ is odd, while it has order $\nicefrac{N}{2}$ when $N$ is even.
If $N$ is odd then $\DD_N$ has two one-dimensional characters, $\chi_1$ and $\chi_2$. If $N$ is even there are two additional one-dimensional characters, $\chi_3$ and $\chi_4$. The values of $\chi_1,\ldots,\chi_4$ are given in Table \ref{tab:one-dimensional characters}.
\begin{table}[h!]
\caption{$\DD_N$ one-dimensional characters: $\chi_1$, $\chi_2$ if $N$ is odd; $\chi_1,\ldots,\chi_4$ if $N$ is even.}
\label{tab:one-dimensional characters}
\begin{center}
\begin{tabular}{c|cccc}
 & $\chi_1$ & $\chi_2$ & $\chi_3$ &$\chi_4$\\
\hline
$r$ & $1$ & $1$ & $-1$ & $-1$\\
$s$ & $1$ & $-1$ & $1$ & $-1$
\end{tabular}
\end{center}
\end{table}

The remaining irreducible characters are two-dimensional. We denote them by $\psi_j$, for $1\le j <\nicefrac{N}{2}$, and they take the values
\[
\psi_j(r^\ell)=\omega^{j\ell}+\omega^{-j\ell},\qquad \psi_j(sr^\ell)=0\,,
\]
where $\omega= e^{\frac{2\pi i}{N}}.$
The characters of the faithful (injective) representations are precisely those $\psi_j$ for which $\textrm{gcd}(j,N)=1$, where $\textrm{gcd}$ stands for \textit{greatest common divisor}. 

When we consider explicit matrices for a representation $\rho:\DD_N\rightarrow \GL(\CC^2)$ with character $\psi_j$, there is an entire equivalence class to choose from. We will consider here the choice
\begin{equation}
\label{eq:standard matrices}
 \rho_r=\begin{mymatrix}\omega^{{j}}&0\\0&\omega^{N-{j}}\end{mymatrix} \,,\qquad \rho_s=\begin{mymatrix}0&1\\1&0\end{mymatrix}\,.
\end{equation}
If one knows the space of invariants $V^G$, when $G$ is represented by $\rho$, then one can easily find the space of invariants belonging to an equivalent representation $\rho'$, because the invertible transformation $T$ that relates the two representations, $T\rho_g=\rho'_g T$, also relates the spaces of invariants: $V^{\rho'(G)}=TV^{\rho(G)}$.

\subsection{Elements of invariant theory}

\begin{definition}
Let $V$ be a vector space and let $\CC[X,Y]$ be the ring of polynomials in $X$ and $Y$; then we define 
\[
\widehat{V}= V\otimes\CC[X,Y]\,.
\]
\end{definition}
We denote, here and in what follows, $\widehat{V}^G$ the space of $G$-invariants in $\widehat{V}$.
An important classical result in invariant theory is Molien's Theorem (see for example \cite{stanley1979invariants,Smith,Neusel}). 
\begin{thm}[Molien]
\label{Molien, thm}
Let $V$ and $U$ be $G$-modules. Then the Poincar\'e series 
\footnote{There is no unanimous consensus regarding the name of the generating function $P\big(M, t\big)$ of a graded module $M$ over a graded algebra: in e.g. \cite{stanley1979invariants}
the term \emph{Hilbert series} is used, while
in e.g. \cite{Smith,Neusel}
 it is called \emph{Poincar\'e series} and we prefer to follow this last convention. 
It is worth noticing that in the context of classical invariant theory, if $M=(V\otimes\CC[U])^G$, it is also called \emph{Molien series}.
}   
for the space of invariants in $V\otimes\CC[U]$ is given by \[P\big((V\otimes\CC[U])^G, t\big)=\frac{1}{|G|}\sum_{g\in G}\frac{\overline{\chi_V(g)}}{\det(1-\sigma_U(g)t)}.\] 
where $\chi_V$ is the character of $V$ and $\sigma_U$ the representation of $U$.
\end{thm}
In this paper $U$ will always be a two-dimensional module with dual basis $\{X,Y\}$.
\begin{example}
[see \cite{Smith}, p.93]
\label{ex:invariant forms}
Consider the cyclic group $\zn{N}$ generated by the matrix $\begin{mymatrix}\omega&0\\0&\omega^{N-1}\end{mymatrix}$ where $\omega$ is a primitive $N$-th root of unity. By direct investigation one finds that $\CC[X,Y]^{\zn{N}}=\bigoplus_{i=0}^{N-1}\CC[X^N,Y^N](XY)^i$ and thus \[P(\CC[X,Y]^{\zn{N}},t)=\frac{\sum_{i=0}^{N-1}t^{2i}}{(1-t^N)^2}=\frac{1-t^{2N}}{(1-t^2)(1-t^N)^2}=\frac{1+t^{N}}{(1-t^2)(1-t^N)}.\]
With this in mind one can determine the invariants for the group $\DD_N$ generated by $\begin{mymatrix}\omega&0\\0&\omega^{N-1}\end{mymatrix}$ and $\begin{mymatrix}0&1\\1&0\end{mymatrix}$,
using Molien's Theorem. Indeed, it follows that
\begin{eqnarray*}
\fl
P(\CC[X,Y]^{\DD_N},t)&=&\frac{1}{2N}\sum_{g\in\DD_N}\frac{1}{\det(1-\sigma_gt)}\\
&=&\frac{1}{2N}\sum_{\ell=0}^{N-1}\left(\left|\begin{array}{cc}1-\omega^\ell t&0\\0&1-\omega^{N-\ell}t\end{array}\right|^{-1}+\left|\begin{array}{cc}1&-\omega^\ell t\\-\omega^{N-\ell}t&1\end{array}\right|^{-1}\right)\\
&=&\frac{1}{2}\left(P(\CC[X,Y]^{\zn{N}},t)+\frac{1}{(1-t^2)}\right)
=\frac{1}{(1-t^2)(1-t^N)}\,.
\end{eqnarray*}
If one can find two algebraically independent $\DD_N$-invariant forms, one of degree $2$ and one of degree $N$, e.g. $XY$ and $X^N+Y^N$, then the above calculation proves that any $\DD_N$-invariant polynomial is a polynomial in these two forms, i.e. $\CC[X,Y]^{\DD_N}=\CC[XY, X^N+Y^N]$.
\end{example}




The $\DD_N$-invariant spaces are known; they are listed in Table \ref{tb:invariants, N odd} and \ref{tb:invariants, N even}, where we use the definitions
\begin{equation}
\label{eq:Fs}
F_0= \frac{X^N-Y^N}{2}\qquad F_1= \frac{X^N+Y^N}{2}\qquad F_2= XY\,,
\end{equation} 
and where $V_\chi$ is the representation space of the character $\chi$.
These forms satisfy one algebraic relation:
\begin{equation}
\label{eq:abg relation}
F_0^2-F_1^2+F_2^N=0\,.
\end{equation}
The appendix contains an explicit calculation for the two-dimensional cases $\widehat{V}_{\psi_j}^{\DD_N}$, for completeness.

\begin{table}[h!]
\caption{Module generators $\eta_i$ in $\widehat{V}_\chi^{\DD_N}=\bigoplus_i\CC[F_1,F_2]\eta_i$, $N$ odd}
\begin{center}
\begin{tabular}{c|ccc}
\label{tb:invariants, N odd}
 & $\chi_1$ & $\chi_2$ & $\psi_j$\\
\hline
$\eta_i$ & $1$ & ${F_0}$ & $\begin{mymatrix}X^j\\Y^j\end{mymatrix},\;\begin{mymatrix}Y^{N-j}\\X^{N-j}\end{mymatrix}$\\
\end{tabular}
\end{center}
\end{table}

\begin{table}[h!]
\caption{Module generators $\eta_i$ in $\widehat{V}_\chi^{\DD_{2N}}=\bigoplus_i\CC[F_1^2,F_2]\eta_i$.}
\begin{center}
\begin{tabular}{c|ccccc}
\label{tb:invariants, N even}
 & $\chi_1$ & $\chi_2$ & $\chi_3$ & $\chi_4$ & $\psi_j$\\
\hline
$\eta_i$ & $1$ & $F_1{F_0}$ & $F_1$ & ${F_0}$ & $\begin{mymatrix}X^j\\Y^j\end{mymatrix},\;\begin{mymatrix}Y^{2N-j}\\X^{2N-j}\end{mymatrix}$\\
\end{tabular}
\end{center}
\end{table}

Let $\chi_V$ denote the character of a $G$-module $V$. Then the character of $\mf{sl}(V)$ is given by
\begin{equation}\label{lem:characters of algebras}
\chi_{\mf{sl}(V)}(g)=\chi_V(g)\overline{\chi_V}(g)-1\qquad \forall g\in G,
\end{equation}
where the overline stands for complex conjugation.
Indeed, $\mf{gl}(V)=\textrm{End}(V)\cong V\otimes V^*$ has character $\chi_V\overline{\chi_V}$, and all scalars $\CC\textrm{Id}$ are invariant.

\begin{example}[$\widehat{\mf{sl}(V_{\psi_j})}^{\DD_N}$]
\label{ex:sl2 invariants}
Using (\ref{lem:characters of algebras}) one finds that 
$\chi_{\mf{sl}(V_{\psi_j})}=\chi_2+\psi_{2j}$. In order to find the invariant matrices explicitly, it is sufficient to identify these two irreducible representations within the space $\mf{sl}(V_{\psi_j})$ generated by $e_0=\begin{mymatrix}1&0\\0&-1\end{mymatrix}$, $e_+=\begin{mymatrix}0&1\\0&0\end{mymatrix}$, and $e_-=\begin{mymatrix}0&0\\1&0\end{mymatrix}$ and then use Table \ref{tb:invariants, N odd} or \ref{tb:invariants, N even}, depending on whether $N$ is odd or even. 

The $\chi_2$ copy in $\chi_{\mf{sl}(V_{\psi_j})}$ is given by $\CC e_0$ and $\{e_+, e_-\}$ is a basis for the $\psi_{2j}$ summand in which $\DD_N$ acts according to our preferred choice of matrices (\ref{eq:standard matrices}). Thus, if e.g. $N$ is odd, the $\DD_N$ invariant algebra $\widehat{\mf{sl}(V_{\psi_j})}^{\DD_N}$ is a free $\CC[F_1, F_2]$-module generated by the matrices
\begin{equation}
\frac{1}{2}\begin{mymatrix} X^N -Y^N& 0 \\ 0 &Y^N-X^N\end{mymatrix}\!,\;
\begin{mymatrix} 0 & X^{2j} \\ Y^{2j} & 0 \end{mymatrix}\!,\; 
\begin{mymatrix} 0 & Y^{N-2j} \\ X^{N-2j} & 0 \end{mymatrix}\!.
\end{equation}
\end{example}

\subsection{Invariant theory on the sphere}

Each noncyclyc finite M\"obius group $G$, such as the dihedral group, is the symmetry group of a tessellation of the sphere with $d_0$ vertices, $d_1$ edges and $d_2$ faces. Let $\Gamma_i\subset\overline{\CC}$ be the set of centroids of $i$-dimensional cells of the tessellation, i.e. an exceptional orbit.

\begin{definition}[The form vanishing on an orbit]
\label{def:form of orbit}
For any $G$-orbit $\Gamma$, we define 
\begin{equation} 
\label{eq:fgamma}
F_\Gamma\in\CC_{|\Gamma|}[X,Y]
\end{equation} as the form vanishing on $\Gamma$.  $F_\Gamma$ is unique up to multiplicative constant. 
Moreover, there is a natural number $\nu_\Gamma$ such that 
\begin{equation*}
\nu_\Gamma |\Gamma|=|G|\,.
\end{equation*}
To shorten notation we use $F_i=F_{\Gamma_i}$ and $\nu_i=\nu_{\Gamma_i}$.
\end{definition}
Under a suitable choice of constants $\sum_{i=0}^{2} F_i^{\nu_i}=0$ and there is no other relation between the forms $F_i$.
Notice that for any form $F$, that is, any homogeneous polynomial of degree $d$, the set of zeros is well defined on the Riemann sphere $\CC P^1$ since $F(cX,cY)=c^dF(X,Y)$ for any scalar $c$.

The forms $F_i$ are relative invariants of a \emph{Schur covering group} $G^\flat$ of $G$ \cite{schur1911darstellung,Curtis}, that is, $g\,F_i=\chi(g) F_i$, $\forall g\in G^\flat$ and where $\chi$ is a one-dimensional character of $G^\flat$. In fact, these forms generate the same ring as the relative invariants of any Schur cover $G^\flat$ do, i.e.
 \[
 \bigslant{\CC[F_0,F_1,F_2]}{F_0^{\nu_0}+F_1^{\nu_1}+F_2^{\nu_2}}=\CC[X,Y]^{[G^\flat,G^\flat]},
 \]
where $[G,G]$ denotes the derived subgroup.
\begin{example}
Let the equator of the Riemann sphere be an $N$-gon with a vertex at $1$. Then the corresponding representation of $\DD_N$ by M\"obius transformations corresponds to the representation $\psi_1$ of $\DD_N^\flat$ in the form (\ref{eq:standard matrices}). It follows that the relative invariants are
\[
\CC[X,Y]^{[\DD_N^\flat,\DD_N^\flat]}=\CC[F_1,F_2]\oplus \CC[F_1,F_2]F_0
\]
where $F_0$, $F_1$ and $F_2$ are given in (\ref{eq:Fs}).
\end{example}

\begin{lem}
\label{lem:finuiinvariant}
Let $F_i$ and $\nu_i$ be defined as above. Then $F_i^{\nu_i}$ is invariant:
\[F_i^{\nu_i}\in \CC_{|G|}[X,Y]^{G^\flat}\,.\]
\end{lem}





The next lemma shows one can find all invariant vectors $(V\otimes\mathcal{M}(\overline{\CC}))^G_{\Gamma}$ by considering quotients of invariant vectors and invariant forms whose degrees are multiples of $|G|$.
\begin{lem}
\label{lem:only quotients of invariants}
Let $V$ and $\overline{\CC}$ be $G$-modules and $\Gamma\subset\overline{\CC}$ a single $G$-orbit. Suppose $\bar{v}\in (V\otimes\mathcal{M}(\overline{\CC}))^G_{\Gamma}$. Then $\bar{v}=\frac{v}{F}$ for an invariant vector $v\in\widehat{V}^{G^\flat}$ and an invariant polynomial $F\in \CC[X,Y]^{G^\flat}$ of the same degree, divisible by $|G|$.
\end{lem}
\begin{proof}
By the identification $\overline{\CC}\cong\CC P^1$, $\bar{v}=\frac{v'}{F'}$ for some $(X,Y)$-dependent vector $v'$ and scalar $F'$ of the same degree. Assume $v'$ and $F'$ have no common factors.

We first argue that $F'$ is a relative invariant of $G^\flat$. Let $g\in G$. Then $g\cdot \bar{v}=g\cdot \frac{v'}{F'}= \frac{g'\cdot v'}{g'\cdot F'}$ for an element in the Schur covering group $g'\in G^\flat$. Hence, by invariance and because $v'$ and $F'$ have no common factor, $g'\cdot v'=Qv'$ and $g'\cdot F'=QF'$ for some
$Q\in\CC$. Moreover, $G^\flat$ is finite so that $Q$ is a root of unity and hence $F'\in\CC[X,Y]^{[G^\flat,G^\flat]}=\bigslant{\CC[F_0,F_1,F_2]}{F_0^{\nu_0}+F_1^{\nu_1}+F_2^{\nu_2}}$. 

If $F'$ vanishes on a generic orbit then its degree is a multiple of $|G|$ and the result follows. If $F'$ vanishes on an exceptional orbit then it is a power of $F_i$ for some $i$.
One can multiply the numerator and denominator of $\bar{v}$ with a power of $F_i$ to get a quotient of invariants and because the degree of $F_i$ divides $|G|$ this can be done ending up with degrees divisible by $|G|$.


\end{proof}

\begin{example}
\label{ex:automorphic functions}
If we apply Lemma \ref{lem:only quotients of invariants} to the trivial $G$-module $V=\CC$, one obtains the automorphic functions of $G$ with poles on a specified orbit. If $f\in\mathcal{M}(\overline{\CC})^G_{\Gamma}$ the invariance and pole restriction imply that the denominator of $f$ is a power of $F_{\Gamma}$. By Lemma \ref{lem:only quotients of invariants} we can assume that it is in fact a power of $F_\Gamma^{\nu_\Gamma}$
and that the numerator is a form in $\CC[F_0^{\nu_0},F_1^{\nu_1},F_2^{\nu_2}]$. The numerator factors into elements of the two-dimensional vector space $\bigslant{\CC F_0^{\nu_0}+\CC F_1^{\nu_1}+\CC F_2^{\nu_2}}{F_0^{\nu_0}+F_1^{\nu_1}+F_2^{\nu_2}}$. One line in this space is $\CC F_\Gamma^{\nu_\Gamma}$. Any vector outside this line will generate the ring of automorphic functions, e.g.
\[
\mathcal{M}(\overline{\CC})^G_{\Gamma}=\CC\left[\frac{F_{i}^{\nu_{i}}}{F_\Gamma^{\nu_\Gamma}}\right]\qquad \Gamma_i\ne\Gamma\,.
\] 
\end{example}

\section{Towards Automorphic Lie Algebras}
\label{sec:tALIAS} 

Our approach to investigate Automorphic Lie Algebras was briefly sketched in the Introduction. In this section we provide the necessary justification.

The field of meromorphic functions on the Riemann sphere $\mathcal{M}(\overline{\CC})$ equals the set of rational functions on $\overline{\CC}$, that is, quotients of polynomials. Using the identification \[\overline{\CC}\ni\lambda=[(X,Y)]\in\CC P^1\] this is the set of quotients of two homogeneous polynomials in $X$ and $Y$ of the same degree. Moreover, M\"obius transformations on $\lambda$ correspond to linear transformation on $(X,Y)$ by the same matrix. However, two matrices yield the same M\"obius transformation if and only if they are scalar multiples of one another. Therefore we allow the action on $(X,Y)$ to be projective in order to cover all possibilities. That is, we require projective representations $\rho:G\rightarrow \textrm{PGL}(\CC^2)$. The same holds for the action on the Lie algebra $\mf{g}(V)<\mf{gl}(V)$.

Projective representations of $G$ are in one to one correspondence with linear representations of a \emph{Schur covering group} $G^\flat$.
The Schur covering groups of the dihedral groups are well known \cite{Re76}.
\begin{prop}
If $N$ is odd then the dihedral group and its Schur covering group coincide: $\DD_N^\flat=\DD_N$. If $N$ is even then $\DD_{2N}$ is one of multiple Schur covering groups for $\DD_N$. 
\end{prop}

We consider a faithful linear representations of $G^\flat$ 
\[\sigma:G^\flat\rightarrow \textrm{GL}(\CC^2)\] 
which restricts $G$ to the groups 
\[\zn{N}\,,\quad\DD_N\,,\quad\TT\,,\quad\OO\,,\quad\YY\] 
of Klein's classification \cite{Klein56,Klein93}, where $\zn{N}$ is the cyclic group, $\DD_N$ the dihedral group, $\TT$ the tetrahedral group, $\OO$ the octahedral group and $\YY$ the icosahedral group.

Let \[\tau:G^\flat\rightarrow\textrm{GL}(V)\] be any irreducible linear representation and consider $\mf{g}(V)$ to be a Lie subalgebra and $G^\flat$-submodule of $\mf{gl}(V)$.
The representations $\sigma$ and $\tau$ induce an action on the space $\mf{g}(V)\otimes\CC[X,Y]$, which concretely reads 
\[g\cdot\big(M\otimes F(X,Y)\big)=\tau_g M \tau_{g^{-1}}\otimes F\big(\sigma_{g^{-1}}(X,Y)\big)
\]
where $M\in \mf{g}(V)$ and $F(X,Y)\in \CC[X,Y]$.
Notice that the space of invariants $(\mf{g}(V)\otimes\CC[X,Y])^{G^\flat}$ is a Lie subalgebra because the $G^\flat$-action respects the Lie bracket.

\begin{definition}[Localisation] Let $F\in\CC_d[X,Y]^G$ and define the multiplicative set $U_F=\bigcup_{\ell\ge0}\CC^*F^\ell$. The space of invariants $\widehat{V}^G$ is a $\CC[X,Y]^G$-module and one can consider its localisation by $U_F$, which we denote \[\widehat{V}^G_F=U_F^{-1}\widehat{V}^G=\left\{\frac{v}{F^\ell}\;|\;v\in \widehat{V}^G,\,\ell\ge 0\right\}\,.\]
\end{definition}
The localisation of a ring by a multiplicative subset $U$ is the smallest extension of the ring in which all elements in $U$ are units, that is, have a multiplicative inverse.
\begin{definition}
Let $G^\flat$ be the Schur cover of $G$; consider the localisation $\widehat{V}^{G^\flat}_{F_{\Gamma}}$. We define $\overline{V}^G_{\Gamma}<\widehat{V}^{G^\flat}_{F_{\Gamma}}$ as the subset of elements that factor through $\CC P^1$
\[
\overline{V}^G_{\Gamma}=\left\{v\in\widehat{V}^{G^\flat}_{F_{\Gamma}}\;|\;v(zX,zY)=v(X,Y),\, \forall z\in\CC^*,\, \forall(X,Y)\in\CC^2\setminus (0,0)\right\}.
\]
\end{definition}
Notice that $v\in\overline{V}^G_{\Gamma}$ has no poles outside $\Gamma$.

Finally, we define two operators, \emph{prehomogenisation} $p$ and \emph{homogenisation} $h$, which will take us from $\widehat{V}^{G^\flat}$ to $\overline{V}^G_{\Gamma}$ in two steps. By taking just the first step, one can study $\overline{V}^G_{\Gamma}$ while holding on to the degree information of the homogeneous polynomials in $(X,Y)$. 

\begin{definition}[Prehomogenisation]
Define $p_d:\widehat{V}^{G}\rightarrow\widehat{V}^{G}$ on basis vectors $v$ of homogeneous degree by \[p_dv=\left\{
\begin{tabular}{ll}
$v$&$\textrm{if }d\,|\deg v$\\
$0$&$\textrm{otherwise}\,.$
\end{tabular}
\right.\]
Then extend linearly.
\end{definition}
\begin{definition}[Homogenisation]
\label{def:hom}
Let $\Gamma\subset\overline{\CC}$ be a $G$-orbit, $F_\Gamma$ defined as in (\ref{eq:fgamma}), and suppose $|\Gamma| \;|\; d$. Define $h_\Gamma:p_d\widehat{V}^{G^\flat}\rightarrow\overline{V}^{G}_\Gamma$ on basis vectors $v$ of homogeneous degree $\ell|\Gamma|$ by
\[h_\Gamma v=\frac{v}{F_\Gamma^\ell}\,.\]
Then extend linearly.
\end{definition}

We define equivalence classes \emph{$\textrm{mod }F$} on $\widehat{V}^G$ by 
\[u\textrm{ mod }F = v\textrm{ mod }F\Leftrightarrow\exists \ell\in\ZZ : u=F^\ell v\,.\]
\begin{lem} Let $\Gamma\subset\overline{\CC}$ be a $G$-orbit, $F_\Gamma$ defined as in (\ref{eq:fgamma}),  and suppose $|\Gamma| \;|\; d$. There is a ring-module isomorphism
\[p_d\widehat{V}^{G^\flat}\textrm{mod }F_\Gamma \cong \overline{V}^G_\Gamma\,.\]
\end{lem}
\begin{proof}
The linear map $h_\Gamma:p_d\widehat{V}^{G^\flat}\rightarrow\overline{V}^{G}_\Gamma$ respects products and is therefore a homomorphism of modules. Its kernel is the equivalence class of the identity and by Lemma \ref{lem:only quotients of invariants} the map is surjective.
\end{proof}
In particular, since the Lie algebra structure of the Automorphic Lie Algebras depends only on the ring structure, that is, addition and multiplication, the previous Lemma also gives Lie algebra isomorphisms.
Notice also that it will always suffice to use $d=|G|$ since the order of any orbit divides the order of the group. Thus we have
\begin{cor} 
Adopting the notation from above, there is a Lie algebra isomorphism
\[p_{|G|}\widehat{\mf{g}(V)}^{G^\flat}\textrm{mod }F_\Gamma\cong \overline{\mf{g}(V)}^G_\Gamma\]
\end{cor}
Because $p_{|G|}\widehat{\mf{g}(V)}^{G^\flat}$ is independent of the orbit $\Gamma$, one can study all Automorphic Lie Algebras $\left\{\overline{\mf{g}(V)}^G_\Gamma\;|\;\Gamma \textrm{ an orbit}\right\}$ by studying  the one Lie algebra $p_{|G|}\widehat{\mf{g}(V)}^{G^\flat}$.

The prehomogenisation projection $p_d$ only concerns degrees. Therefore it makes sense to define it on Poincar\'e series. Since there is little opportunity for confusion, we use the same name:
\[p_d:\ZZ_{\ge0}[t]\rightarrow\ZZ_{\ge0}[t]\,,\quad \sum_{\ell\ge0}k_\ell t^\ell\mapsto\sum_{d\,|\,\ell}k_\ell t^\ell\,.\]

\begin{example}[$\overline{V_{\chi_1}}^{\DD_N}_{\Gamma}$]
\label{ex:automorphic forms} The case of the trivial representation $V_{\chi_1}$ corresponds to automorphic functions.
First we assume that $N$ is odd so that $\DD_N^\flat=\DD_N$ and $\widehat{V}_{\chi_1}^{\DD_N^\flat}=\widehat{V}_{\chi_1}^{\DD_N}=\CC[F_1,F_2]$. Recall that $\deg F_1=N$ and $\deg F_2=2$ and the two forms are algebraically independent. 
The prehomogenisation projection $p_{2N}$ maps
\[
\fl
P(\widehat{V}_{\chi_1}^{\DD_N^\flat},t)=\frac{1}{(1-t^2)(1-t^N)}=\frac{(1+t^2+t^4+\ldots+t^{2N-2})(1+t^N)}{(1-t^{2N})^2}\mapsto\frac{1}{(1-t^{2N})^2}
\]
and one finds \[p_{2N}\CC[F_1,F_2]=\CC[F_1^{\nu_1},F_2^{\nu_2}]\,.\]
If $N$ is even we use the Schur cover $\DD_{2N}$ with invariant forms $\widehat{V}_{\chi_1}^{\DD_{2N}}=\CC[F_1^{2},F_2]$, which are mapped onto the same ring under $p_{2N}$.

If the equivalence relation $\textrm{mod }F_\Gamma$ is introduced one is left with $\CC[F_k^{\nu_k}]$ for $F_k\ne F_\Gamma$.
Notice that this ring is equivalent to $\mathcal{M}(\overline{\CC})^G_{\Gamma}=\CC\left[\frac{F_{k}^{\nu_{k}}}{F_\Gamma^{\nu_\Gamma}}\right]$, $\Gamma_k\ne\Gamma$ from Example \ref{ex:automorphic functions}.

\end{example}

\section{Automorphic Lie Algebras with dihedral symmetry}
\label{sec:D_NALIAS} 

Recall from Example \ref{ex:sl2 invariants} that \[\widehat{\mf{sl}(V_{\psi_j})}^{\DD_N}=\CC[{F_1},{F_2}]\eta_1\oplus\CC[{F_1},{F_2}]\eta_2\oplus\CC[{F_1},{F_2}]\eta_3\] 
where, in our preferred basis (\ref{eq:standard matrices}), $F_0$, $F_1$ and $F_2$ are given in (\ref{eq:Fs}) and where
\begin{center}
\begin{tabular}{lcr}
$\eta_1=  \begin{mymatrix} 0 & X^{2j} \\ Y^{2j} & 0 \end{mymatrix}\,,$ &
$\eta_2 =\begin{mymatrix} 0 & Y^{N-2j} \\ X^{N-2j} & 0 \end{mymatrix}\,,$ & 
$\eta_3 =\nicefrac{1}{2}(X^N-Y^N) \begin{mymatrix}1 & 0 \\ 0 &-1 \end{mymatrix}\,.$
\end{tabular}
\end{center}
The matrices $\eta_1$, $\eta_2$ and $\eta_3$ generate a Lie algebra over $\CC[F_1, F_2]$, with commutator brackets
\begin{eqnarray*}
&& [\eta_2,\eta_3]=2(-{F_2}^{N-2j}\eta_1+{F_1} \eta_2)\\
&& [\eta_3,\eta_1]=2({F_1} \eta_1-{F_2}^{2j} \eta_2)\\
&&[\eta_1,\eta_2]=2\eta_3\,.
\end{eqnarray*}
Notice that this is not an Automorphic Lie Algebra yet, 
but it can be made into one using the homogenisation operator $h$ (see Definition \ref{def:hom}). 
\begin{example} 
It is useful to compare ours with previous results, starting from this algebra. In particular, let us consider \cite{LM05comm} and \cite{Bury2010}, which contain explicit descriptions of Automorphic Lie Algebras with dihedral symmetry with poles restricted to the orbit of two points, which we call $\Gamma_2=\{0,\infty\}$.

The generators (27) in \cite{LM05comm}, page 190, where $N=2$ and $j=1$ 
are nothing but
$h_{\Gamma_2}(\eta_1)$, $h_{\Gamma_2}(\eta_2)$ and $2h_{\Gamma_2}(\eta_3)$, respectively.

Similarly,  generators (4.10) in \cite{Bury2010}, page 81, for general $N$ and $j$, are $h_{\Gamma_2}(\eta_1)$, $h_{\Gamma_2}(2\eta_2-\eta_1)$ and $4h_{\Gamma_2}(\eta_3)$ respectively, if $N$ is odd, and $h_{\Gamma_2}(\eta_1)$, $h_{\Gamma_2}(\eta_2)$ and $2h_{\Gamma_2}(\eta_3)$ respectively, if $N$ is even.
%
%
\end{example}

Observe that he image of $\widehat{\mf{sl}(V_{\psi_j})}^{\DD_N^\flat}$ under $p_{|\DD_N|}=p_{2N}$ is 
\[\CC[F_1^2,F_2^N]\left(F_2^{N-j}\eta_1\oplus F_1F_2^j\eta_2\oplus F_1\eta_3\right)\]
regardless of whether $N$ is odd, in which case $\DD_N^\flat =\DD_N$, or even, in which case we use $\DD_N^\flat =\DD_{2N}$. This reflects the fact that $p_{|G|}\widehat{\mf{sl}(V_{\psi_j})}^{G^\flat}\textrm{mod }F_\Gamma \cong \overline{\mf{sl}(V_{\psi_j})}^{G}_\Gamma$ is a space of $G$-invariants rather than $G^\flat$-invariants (see also Example \ref{ex:automorphic forms}). 

The module generators of $p_{|\DD_N|}\widehat{\mf{sl}(V_{\psi_j})}^{\DD_N^\flat}$ will be denoted with a tilde:
\begin{eqnarray*}
\tilde{\eta}_1=F_2^{N-j}\eta_1=F_2^N\begin{mymatrix}0&\lambda^j\\\lambda^{-j}&0\end{mymatrix}\,;\\
\tilde{\eta}_2=F_1F_2^j\eta_2=F_2^N\begin{mymatrix}0&\nicefrac{1}{2}(\lambda^j+\lambda^{j-N})\\\nicefrac{1}{2}(\lambda^{-j}+\lambda^{N-j})&0\end{mymatrix}\,;\\
\tilde{\eta}_3=F_1\eta_3=F_0F_1\begin{mymatrix}1&0\\0&-1\end{mymatrix}\,.
\end{eqnarray*}
One readily computes the structure constants
\begin{eqnarray*}
&& [\tilde{\eta}_2,\tilde{\eta}_3]=2(-{F_1}^{2}\tilde{\eta}_1+{F_1}^2 \tilde{\eta}_2)\\
&& [\tilde{\eta}_3,\tilde{\eta}_1]=2({F_1}^2 \tilde{\eta}_1-{F_2}^{N} \tilde{\eta}_2)\\
&&[\tilde{\eta}_1,\tilde{\eta}_2]=2F_2^N\tilde{\eta}_3\,.
\end{eqnarray*}
At this stage, one could apply the homogenisation operator $h$ to obtain the wanted Automorphic Lie Algebra. However, before doing so, we prefer to define a \emph{normal form} notion.

\subsection{A normal form for Automorphic Lie Algebras}
\label{sec:nf}

In order to compare Automorphic Lie Algebras with each other it is useful to define a normal form. Once two algebras are transformed to a particular normal form one can immediately see whether they are isomorphic.
 We 
 define a normal form for Automorphic Lie Algebras as similar to the Serre normal form for simple Lie algebras as possible, and we will use here the same name.
\begin{definition}[Serre normal form]
Let $\mf{g}(V)$ be a simple Lie algebra with Cartan subalgebra (CSA) $\mf{h}$. We say that the module generators of an Automorphic Lie Algebra $\overline{\mf{g}(V)}^G_\Gamma$ are in Serre normal form 
if the adjoint action of $\dim_\CC\mf{h}$ of the $\dim_\CC\mf{g}(V)$ generators is diagonal on all generators, with the same roots as $\mf{h}$ in Serre normal form.
\end{definition}
This definition relies on the fact that the Automorphic Lie Algebra $\overline{\mf{g}(V)}^G_\Gamma$ can be generated by  $\dim \mf{g}(V)$ elements. In the previous section we saw that this is the case if $G=\DD_N$: we found the $\tilde{\eta}_i$. In \cite{KLS_2014_SLN} it is shown that this holds for all the polyhedral groups.

Despite this fact, it is not obvious to us for which Automorphic Lie Algebras a Serre normal form exists. Moreover, we find that the normal form depends on the Automorphic Lie Algebra CSA; this is, to the best of our knowledge, a new phenomenon in Lie algebra theory, and it deserves further investigation. The fact that we work over a ring complicates matters compared to the classical situation. The main result of this paper describes the Lie algebra structure of Automorphic Lie Algebras with dihedral symmetry by a construction of this normal form. We reformulate this theorem, which was already sketched in the Introduction, as follows:
\begin{thm}
\label{thm:main2}

Let $\DD_N$ act faithfully on the Riemann sphere and let $\Gamma$ be a single orbit therein.
Let $(c_1,c_2)\in\CC P^1$ be such that \[F=c_1F_1^2+c_2F_2^N\] vanishes on $\Gamma$, if $F_i$ is given by (\ref{eq:Fs}). 

Then, adopting the previous notation, for $\Gamma_k\ne \Gamma$,
\[\overline{\mf{sl}(V_{\psi_j})}^{\DD_N}_{\Gamma}=\left(e_0\oplus e_+\oplus e_-\right)\otimes\CC\left[\frac{F_k^{\nu_k}}{F}\right]\]
and 
\begin{eqnarray*}
&&[e_0,e_\pm]=\pm 2e_\pm 
\\&&[e_+,e_-]=e_0\prod_{i=0}^2\frac{F_i^{\nu_i}}{F}\,.
\end{eqnarray*}
A set of module-generators for this normal form is given by
\begin{eqnarray*}
\fl
e_0=\frac{1}{F}\begin{mymatrix}
c_1F_0F_1 & F_2^N(c_2\lambda^j+\nicefrac{c_1}{2}(\lambda^j+\lambda^{j-N}))\\
F_2^N(c_2\lambda^{-j}+\nicefrac{c_1}{2}(\lambda^{-j}+\lambda^{N-j})) & -c_1F_0F_1
\end{mymatrix}\\
\fl
e_+=\frac{F_0F_1F_2^N}{2F^2}\begin{mymatrix}
1&-\lambda^j\\
\lambda^{-j}&-1
\end{mymatrix}\\
\fl
e_-
=\frac{F_0F_1F_2^N}{2F^2}
\begin{mymatrix}
c_1+c_2+c_1c_2\frac{F_0^2}{F}&\left(c_2(c_1+c_2)\frac{F_2^N}{F}+c_1\lambda^{-N}\right)\lambda^{j}\\
-\left(c_2(c_1+c_2)\frac{F_2^N}{F}+c_1\lambda^N\right)\lambda^{-j}& -c_1-c_2-c_1c_2\frac{F_0^2}{F}
\end{mymatrix}\,.
\end{eqnarray*}
\end{thm}
\begin{proof}
We prove the theorem by an explicit construction.
Consider the prehomogenised Lie algebra $p_{|2N|}\widehat{\mf{sl}(V_{\psi_j})}^{\DD_N^\flat}=\left(\tilde{\eta}_1\oplus\tilde{\eta}_2\oplus\tilde{\eta}_3\right)\CC[F_1^2,F_2^N]$. Define $(\tilde{e}_0,\tilde{e}_+,\tilde{e}_-)$ by \[(\tilde{e}_0,\tilde{e}_+,\tilde{e}_-)=(\tilde{\eta}_1,\tilde{\eta}_2,\tilde{\eta}_3)T\] where 
\[T=\frac{1}{2}\begin{mymatrix}
2c_2&-F_1^2&c_2^2F_0^2F_2^N+F^2\\
2c_1&F_2^N&-c_1^2F_0^2F_1^2-F^2\\
2c_1&F_2^N&F_2^N(c_1c_2F_0^2+(c_1+c_2)F)
\end{mymatrix}\,.\]
It is a straightforward though very tedious exercise to check that this transformation has determinant $\nicefrac{1}{2}F^3$, thus is invertible in the localised ring, and that
\begin{eqnarray*}
&&[\tilde{e}_0,\tilde{e}_\pm]=\pm 2F\tilde{e}_\pm 
\\&&[\tilde{e}_+,\tilde{e}_-]=FF_0^2F_1^2F_2^N \tilde{e}_0 \,.
\end{eqnarray*}
Hence the homogenised matrices $e_{\cdot}=h_\Gamma(\tilde{e}_{\cdot})$ satisfy the normal form described in the theorem.
\end{proof}
We observe that different choices of transformation groups have been made in previous works. Here we follow \cite{LS10} and allow invertible transformations $T$ on the generators over the ring of invariant forms, contrary to \cite{Bury2010}, where only linear transformations over $\CC$ are considered. 

To find the desired isomorphism $T$ one first looks for a matrix $\tilde{e}_0\in \left(\tilde{\eta}_1\oplus\tilde{\eta}_2\oplus\tilde{\eta}_3\right)\CC[F_1^2,F_2^N]$ which has eigenvalues $\pm\mu$ that are  units in the localised ring, that is, powers of $F$. This yields the first column of $T$. The other two columns are found by diagonalising $\textrm{ad}(\tilde{e}_0)$, i.e. solving $[\tilde{e}_0,\tilde{e}_\pm]=\pm 2 \mu \tilde{e}_\pm$ for $\tilde{e}_\pm\in\left(\tilde{\eta}_1\oplus\tilde{\eta}_2\oplus\tilde{\eta}_3\right)\CC[F_1^2,F_2^N]$.

In Table \ref{tb:invariant matrices with exceptional poles} we present the invariant matrices when $\Gamma$ is one of the exceptional orbits. In other words, we evaluate the matrices in Theorem $\ref{thm:main2}$ in $(c_1,c_2)=(1,-1)$ for poles at $\Gamma_0$, in $(c_1,c_2)=(1,0)$ for poles at $\Gamma_1$ and in $(c_1,c_2)=(0,1)$ for poles at $\Gamma_2$.

\begin{table}[h!]
\caption{Generators in normal form of Automorphic Lie Algebras with dihedral symmetry and exceptional poles}
\begin{center}
\scalebox{0.965}{
\begin{tabular}{lccc}
\label{tb:invariant matrices with exceptional poles}
 $\Gamma$ & $e_0$ & $e_+$ & $e_-$ \\
\hline
$\Gamma_0$ &
$\frac{1}{\lambda^N-1}\begin{mymatrix}\lambda^N+1&-2\lambda^j\\2\lambda^{N-j}&-\lambda^N-1\end{mymatrix}$ & 
$\frac{2(\lambda^N+1)\lambda^N}{(\lambda^N-1)^3}\begin{mymatrix}1&-\lambda^j\\\lambda^{-j}&-1\end{mymatrix}$ & 
$\frac{2(\lambda^N+1)\lambda^N}{(\lambda^N-1)^3}\begin{mymatrix}-1&\lambda^{j-N}\\-\lambda^{N-j}&1\end{mymatrix}$ \\
$\Gamma_1$ &
$\frac{1}{\lambda^N+1}\begin{mymatrix}\lambda^N-1&-2\lambda^j\\2\lambda^{N-j}&-\lambda^N+1\end{mymatrix}$ & 
$\frac{2(\lambda^N-1)\lambda^N}{(\lambda^N+1)^3}\begin{mymatrix}1&-\lambda^j\\\lambda^{-j}&-1\end{mymatrix}$ & 
$\frac{2(\lambda^N-1)\lambda^N}{(\lambda^N+1)^3}\begin{mymatrix}1&\lambda^{j-N}\\-\lambda^{N-j}&-1\end{mymatrix}$ \\
$\Gamma_2$ &
$\begin{mymatrix}0&\lambda^j\\\lambda^{-j}&0\end{mymatrix}$ & 
$\frac{\lambda^{2N}-1}{8\lambda^N}\begin{mymatrix}1&-\lambda^j\\\lambda^{-j}&-1\end{mymatrix}$ & 
$\frac{\lambda^{2N}-1}{8\lambda^N}\begin{mymatrix}1&\lambda^j\\-\lambda^{-j}&-1\end{mymatrix}$ \\
\end{tabular}
}
\end{center}
\end{table}

The three generators given in Theorem $\ref{thm:main2}$ can also be presented explicitly in $\lambda$ as
\begin{eqnarray*}
\fl
e_0=\frac{1}{\frac{c_1}{4}(\lambda^N+1)^2+c_2\lambda^N}\begin{mymatrix}
\frac{c_1}{4}(\lambda^{2N}-1) & \lambda^N(c_2\lambda^j+\nicefrac{c_1}{2}(\lambda^j+\lambda^{j-N}))\\
\lambda^N(c_2\lambda^{-j}+\nicefrac{c_1}{2}(\lambda^{-j}+\lambda^{N-j})) & \frac{c_1}{4}(1-\lambda^{2N})
\end{mymatrix}\\
\fl
e_+=\frac{\frac{1}{4}\lambda^N(\lambda^{2N}-1)}{2(\frac{c_1}{4}(\lambda^N+1)^2+c_2\lambda^N)^2}\begin{mymatrix}
1&-\lambda^j\\
\lambda^{-j}&-1
\end{mymatrix}\\
\fl
e_-=\frac{\frac{1}{4}\lambda^N(\lambda^{2N}-1)}{2(\frac{c_1}{4}(\lambda^N+1)^2+c_2\lambda^N)^2}\\
\fl 
\times
\begin{mymatrix}
c_1+c_2+\frac{\frac{c_1c_2}{4}(\lambda^N-1)^2}{\frac{c_1}{4}(\lambda^N+1)^2+c_2\lambda^N}
&c_2(c_1+c_2)\frac{\lambda^{N+j}}{\frac{c_1}{4}(\lambda^N+1)^2+c_2\lambda^N}+c_1\lambda^{-N+j}\\
-c_2(c_1+c_2)\frac{\lambda^{N-j}}{\frac{c_1}{4}(\lambda^N+1)^2+c_2\lambda^{N}}-c_1\lambda^{N-j}
&-c_1-c_2-\frac{\frac{c_1c_2}{4}(\lambda^N-1)^2}{\frac{c_1}{4}(\lambda^N+1)^2+c_2\lambda^N}
\end{mymatrix}.
\end{eqnarray*}

Theorem \ref{thm:main2} describes Automorphic Lie Algebras with dihedral symmetry, with poles at any orbit. Its proof exhibits the consequence of Lemma \ref{lem:only quotients of invariants} that one can compute all these algebra in one go. In particular, one does not need to distinguish between exceptional orbits and generic orbits. 

The resulting Lie algebras differ only in the bracket $[e_+,e_-]=\frac{F_0^{\nu_0}}{F}\frac{F_1^{\nu_1}}{F}\frac{F_2^{\nu_2}}{F} e_0$. Here one can discriminate between the case of generic and exceptional orbits since precisely one factor $\frac{F_i^{\nu_i}}{F}$ equals $1$ if and only if $\Gamma$ is an exceptional orbit. In other words, $\frac{F_0^{\nu_0}}{F}\frac{F_1^{\nu_1}}{F}\frac{F_2^{\nu_2}}{F}$ is a polynomial in $\frac{F_k^{\nu_k}}{F}$, $\Gamma_k\ne\Gamma$, of degree $2$ or $3$ if $\Gamma$ is exceptional or generic respectively. Notice that this degree equals the complex dimension of the quotient 
\[\bigslant{\overline{\mf{sl}(V_{\psi_j})}^{\DD_N}_{\Gamma}}{\left[\overline{\mf{sl}(V_{\psi_j})}^{\DD_N}_{\Gamma},\overline{\mf{sl}(V_{\psi_j})}^{\DD_N}_{\Gamma}\right]}\cong \bigslant{\CC\left[\frac{F_k^{\nu_k}}{F}\right]}{\CC\left[\frac{F_k^{\nu_k}}{F}\right]\frac{F_0^{\nu_0}}{F}\frac{F_1^{\nu_1}}{F}\frac{F_2^{\nu_2}}{F} }\]
in agreement with results in \cite{Bury2010}.

Another way to analyze the Automorphic Lie Algebras
is by considering their values at a particular point $\lambda$. This can be done in the Lie algebra, i.e. the strucure constants, or in the representation of the Lie algebra, i.e. the invariant matrices.
In the first case we have a three-dimensional Lie algebra, equivalent to $\mf{sl}_2(\CC)$ if $\lambda\in \overline{\CC}\setminus (\Gamma\cup\Gamma_0\cup\Gamma_1\cup\Gamma_2)$. When on the other hand $\lambda\in \Gamma_i\ne\Gamma$ we obtain the Lie algebra 
\[
[e_0,e_\pm]=\pm 2 e_\pm\,,\quad [e_+,e_-]=0 \,.
\]

Evaluating the invariant matrices in $\lambda\in \overline{\CC}\setminus (\Gamma\cup\Gamma_0\cup\Gamma_1\cup\Gamma_2)$ also results in $\mf{sl}_2(\CC)$. If, on the other hand, $\lambda\in\Gamma_i\ne\Gamma$ then two generators $e_\pm$ vanish and one is left with a one-dimensional, and in particular commutative, Lie algebra. For completeness, we mention that if $\lambda\in\Gamma$ then all generators of $\overline{\mf{sl}(2,\CC)}^{\DD_N}_\Gamma(\lambda)$ contain clearly singularities, by construction.

\subsection{Other base Lie algebras}

From the offset of this paper we have restricted our attention to simple Lie algebras $\mf{g}(V)$, leaving in the dihedral case only $\mf{sl}_2(\CC)$. It is however straightforward if not particularly interesting to consider other possibilities, hence we do this here for completeness.

Let $V$ be an irreducible $\DD_N$-module. We are interested in subspaces $\mf{g}(V)<\mf{gl}(V)$ which are both a Lie algebra and a $\DD_N$-module. 
The dimension of $V$ is $1$ or $2$. In the first case the action on $\mf{gl}(V)$ is trivial and the Automorphic Lie Algebras $\overline{\mf{g}(V)}^{G}_{\Gamma}=\mathcal{M}(\overline{\CC})^G_\Gamma$ are the rings of automorphic functions analytic outside $\Gamma$. 

Let $V$ now be a two-dimensional $\DD_N$-module affording the character $\psi_j$. Then the $\DD_N$-module $\mf{gl}(V)$ has character $\psi_j^2=\chi_1+\chi_2+\psi_{2j}$ and all $\DD_N$-submodules $\mf{g}(V)<\mf{gl}(V)$ are given by a subset of these characters. Now we add the condition that $\mf{g}(V)$ be a Lie algebra. To this end we look at the matrices affording the available characters $\chi_1+\chi_2+\psi_{2j}$ as in Example \ref{ex:sl2 invariants}, where we noticed that the bracket of the $\psi_{2j}$ summand equals the $\chi_2$ summand. Any Lie algebra containing the first therefore contains the latter. This turns out to be the only restriction on the submodule $\mf{g}(V)$ coming from the requirement that it be a Lie algebra. Lie algebras of all dimensions $\le 4$ are available. The only noncommutative ones are $\mf{sl}(V)$, affording $\chi_2+\psi_{2j}$, and $\mf{gl}(V)$. 
Since $\overline{\mf{gl}(V)}^G_\Gamma=\mathcal{M}(\overline{\CC})^G_\Gamma\textrm{Id}\oplus\overline{\mf{sl}(V)}^G_\Gamma$ considering $\mf{sl}(V)$ is rather natural.

\subsection{Matrices of invariants}
\label{sec:moi}
We conclude this section with yet another representation of the algebra. 
Invariant matrices act on invariant vectors by multiplication. The description of the invariant matrices in terms of this action yields greatly simplified matrices, which we will call \emph{matrices of invariants}, while preserving the structure constants of the Lie algebra. The entries of these matrices are indeed invariant, but the matrices are not invariant under the group action. 


There are different ways to compute matrices of invariants. Here, for example, we choose two vectors, $v_1$ and $v_2$, satisfying three conditions: they are invariant, independent eigenvectors of the matrices in the Cartan subalgebra of the Automorphic Lie Algebras (in this case only $e_0$) and the difference of their degrees is a multiple of $|G|=2N$. This will ensure that the transformed matrices have entries in the ring $p_{|G|}\CC[X,Y]^{G^\flat}=\CC[F_1^2,F_2^N]$ and the matrices in the Cartan subalgebra will be diagonal. Alternatively, one could look at the Molien series for the equivariant vectors and consider the lowest degree at which there are two different equivariant vectors, as a guiding principle.

Let us choose 
\[v_1=\begin{mymatrix}X^j\\Y^j\end{mymatrix},
\qquad v_2=(c_1+c_2)F_1^2F_2^N\begin{mymatrix}X^j\\Y^j\end{mymatrix}-FF_1 F_2^j\begin{mymatrix}Y^{N-j}\\X^{N-j}\end{mymatrix}.\]
resulting in
\[
e_0\sim\begin{mymatrix}1&0\\0&-1\end{mymatrix},\qquad
e_+\sim \begin{mymatrix}0&\frac{F_0^{\nu_0}F_1^{\nu_1}F_2^{\nu_2}}{F^{3}}\\0&0\end{mymatrix},\qquad
e_-\sim \begin{mymatrix}0&0\\1&0\end{mymatrix}.
\]

This approach is especially useful when investigating Automorphic Lie Algebras with  $\TT$, $\OO$ or $\YY$ symmetry, because then the invariant matrices do not fit a page
and in the worst case it takes a top of the range workstation months
to calculate the structure constants. However, if one can construct the matrices of invariants first, it is easy to find the structure constants.

\section{Conclusions}
\label{sec:conc}
Among Automorphic Lie Algebras those with dihedral symmetry have been studied the most (e.g.  \cite{LM04iop,Mikhailov81,Wa09,konstantinou2013darboux,MPW2014}) and, together with icosahedral symmetry, dihedral symmetry seems the most relevant in physical systems. From a mathematical point of view,
$\DD_N$ is the only non-abelian group in Klein's classification \cite{Klein56,Klein93}, whose order depends on $N$. On the classification side, many results are available, starting with the seminal results in \cite{Mikhailov81}; the first attempt towards a systematic classification can be found in \cite{LM04iop}, further development in \cite{LM05comm,LS10,Bury2010}. 


In this paper we provide a complete classification of 
Automorphic Lie Algebras with dihedral symmetry, with poles at any orbit (see Theorem \ref{thm:main2}). 
We find that  there exists a \emph{normal form} for the algebras given by 
\begin{eqnarray*}
&&[e_0,e_\pm]=\pm 2e_\pm 
\\&&[e_+,e_-]=\frac{F_0^{\nu_0}}{F}\frac{F_1^{\nu_1}}{F}\frac{F_2^{\nu_2}}{F} e_0
\end{eqnarray*}
where $\{e_0\,,\,e_\pm\}$ is a set of generators, $F=c_1F_1^2+c_2F_2^N$ is a nonzero invariant form vanishing on $\Gamma$, $F_i$ are the relative automorphic functions given by (\ref{eq:Fs}), after identification of the Riemann sphere with the complex projective line $\CC P^1$ consisting of quotients $\frac{X}{Y}$ of two complex variables by setting $\lambda=\frac{X}{Y}$, and where $\nu_i$ are given in Definition \ref{def:form of orbit}.
It is worth pointing out that the symmetric Lie algebras differ only in the bracket $[e_+,e_-]=\frac{F_0^{\nu_0}}{F}\frac{F_1^{\nu_1}}{F}\frac{F_2^{\nu_2}}{F}e_0 $. Thus one can discriminate between the case of generic and exceptional orbits since precisely one factor $\frac{F_i^{\nu_i}}{F}$ equals $1$ if and only if $\Gamma$ is an exceptional orbit. \\
Theorem \ref{thm:main2} is an important step towards the complete classification of Automorphic Lie Algebras, as the simplifying assumption made in \cite{LS10} to use the same matrices representing the reduction group to act on the spectral parameter as well as on the base Lie algebra is dropped. This simplifying assumption is no longer valid when considering higher dimensional Lie algebras, so the result is also a step forward in the classification beyond $\mathfrak{sl}_2(\CC)$.
Moreover, there is no need to distinguish between \emph{generic} and \emph{exceptional} orbits, as they can be all treated in the same way.
The classification is thus \emph{uniform}, both in the choice of representations and in the choice of the orbits.\\ 
The  introduction of a normal form is also a fundamental step, even more so in higher dimensional cases, as it allows to consider whether algebras are isomorphic or not. 
Preliminary results show that having a normal form makes the problem feasible. 
Similarly, the classification of dihedral Automorphic Lie Algebras associated to \emph{reducible} $\DD_N$-representations and any simple (or semisimple) Lie algebra $\mathfrak{g}$ seems within reach.

We also introduce the concept of \emph{matrices of invariants} (see Section \ref{sec:moi}); they describe the (multiplicative) action of invariant matrices on invariant vectors. The description of the invariant matrices in terms of this action yields greatly simplified matrices, while preserving the structure constants of the Lie algebra. 
We believe that the matrices of invariants will play a fundamental role in the classification of higher dimensional Automorphic Lie Algebras. 

\appendix
\section{Invariant vectors}
\label{sec:invariant_vectors}

\begin{thm}[Invariant vectors]
\label{thm:invariant vectors}
Let $\DD_N$ act on $(X,Y)$ with character $\psi_1$. In the basis corresponding to (\ref{eq:standard matrices}), one finds the relative invariant forms $F_i$ as in (\ref{eq:Fs}) and the space of invariant vectors
\[
\big(V_{\psi_j}\otimes \CC[X,Y]\big)^{\DD_N}=\left( \begin{mymatrix} X^{j}\\Y^{j} \end{mymatrix}\oplus \begin{mymatrix}Y^{N-j}\\X^{N-j}\end{mymatrix}\right) \otimes\CC[F_1,F_2]
\]
where the sum is direct over the ring $\CC[F_1,F_2]$.
\end{thm}

\begin{proof}
An object is invariant under a group action if it is invariant under the action of all generators of a group. To find all $\DD_N=\langle r, s \rangle$- invariant vectors we first look for the $\langle r \rangle=\zn{N}$- invariant vectors and then average over the action of $s$ to obtain all dihedral invariant vectors.

The space of invariant vectors is a module over the ring of invariant forms. When searching for $\zn{N}$-invariant vectors, one can therefore look for invariants modulo powers of the $\zn{N}$-invariant forms $XY$, $X^N$ and $Y^N$. Moreover, we represent $\zn{N}$ by diagonal matrices (which is possible because $\zn{N}$ is abelian). Hence, if $\zn{N}$ acts on $V=\langle e_1, e_2\rangle$ then $\zn{N} e_i\subset \langle e_i \rangle$. Therefore, one only needs to investigate the vectors $X^d e_i$ and $Y^{d} e_i$ for $d \in\{0,\ldots,N-1\}$ and $i\in\{1,2\}$.

Let $\sigma$ be the action on $(X,Y)$ and $\tau$ the action on the vectors. We use the basis in which $\sigma_r= \begin{mymatrix}\omega&0\\0&\omega^{N-1}\end{mymatrix}$ and $\tau_r= \begin{mymatrix}\omega^{{j}}&0\\0&\omega^{N-{j}}\end{mymatrix}$, 
so that \[r X^d e_1 = \omega^{{j}} \omega^{-d} X^d e_1.\] We want to solve $\omega^{{j}} \omega^{-d}=1$, i.e. $d-j\in N\ZZ$. Hence $d\in(j+N\ZZ)\cap\{0,\ldots,N-1\}=j$. That is, $X^{j} e_1$ is invariant under the action of $\langle r \rangle\cong \zn{N}$.

Let us now consider the next one, $rY^{d} e_1=\omega^{j}\omega^{d}Y^{d} e_1$. We solve $\omega^{{j}} \omega^{d}=1$ i.e. $d+j \in N\ZZ$. This implies $d\in (N-j+N\ZZ)\cap \{0, \ldots, N-1\}=N-j$, thus $rY^{N-j} e_1=Y^{N-j} e_1$ is invariant.

Similarly one finds the invariant vectors $Y^{j} e_2$ and $X^{N-j} e_2$,
resulting in the space 
\begin{eqnarray*} 
\fl
\left(\CC^2\otimes\CC[X,Y]\right)^{\zn{N}}
=\left( \begin{mymatrix} X^{j}\\0 \end{mymatrix} + \begin{mymatrix} 0\\Y^{j} \end{mymatrix} + \begin{mymatrix} Y^{N-j}\\0 \end{mymatrix} + \begin{mymatrix}0\\X^{N-j}\end{mymatrix}\right) \otimes\CC[X,Y]^{\zn{N}}\,.
\end{eqnarray*}
If we use the fact $\CC[X,Y]^{\zn{N}}=(1\oplus F_0)\otimes\CC[F_1, F_2]$ this space is generated as a $\CC[F_1, F_2]$ module by the vectors
{\small \begin{eqnarray*}
\fl
\begin{mymatrix} X^{j}\\0 \end{mymatrix}\!, \begin{mymatrix} 0\\Y^{j} \end{mymatrix}\!, \begin{mymatrix} Y^{N-j}\\0 \end{mymatrix}\!, \begin{mymatrix}0\\X^{N-j}\end{mymatrix}\!,\,
F_0\begin{mymatrix} X^{j}\\0 \end{mymatrix}\!, F_0\begin{mymatrix} 0\\Y^{j} \end{mymatrix}\!, F_0\begin{mymatrix} Y^{N-j}\\0 \end{mymatrix}\!, F_0\begin{mymatrix}0\\X^{N-j}\end{mymatrix}\!.
\end{eqnarray*}
}

It turns out the above vectors are dependent over the ring $\CC[F_1, F_2]$. One finds that 
\[F_0\begin{mymatrix} X^{j}\\0 \end{mymatrix}=f(F_1,F_2)\begin{mymatrix} X^{j}\\0 \end{mymatrix}+g(F_1,F_2)\begin{mymatrix} Y^{N-j}\\0 \end{mymatrix}\] 
if $f(F_1,F_2)=F_1$ and $g(F_1,F_2)=-F_2^{j}$, hence this vector is redundant. Similarly, the other vectors with a factor $F_0$ can be expressed in terms of the vectors without this factor.

The remaining vectors between the brackets are independent over the ring $\CC[F_1, F_2]$. Indeed, let $f,g\in\CC[F_1, F_2]$ and consider the equation 
\[fX^{j}+gY^{N-j}=0\,.\] 
If the equation is multiplied by $Y^{j}$ we find 
\[fF_2^{j} +gY^{N}=fF_2^{j} +g (F_1-F_0)=0\,.\] Now one can use the fact that all terms are invariant under the action of $s$ except for $gF_0$ to see that $g=0$, and hence $f=0$.  Similarly $fY^{j}+gX^{N-j}=0$ implies $f=g=0$. Thus we have a direct sum
\[
\fl
\left(\CC^2\otimes\CC[X,Y]\right)^{\zn{N}}
=\left( \begin{mymatrix} X^{j}\\0 \end{mymatrix}\oplus \begin{mymatrix} 0\\Y^{j} \end{mymatrix}\oplus \begin{mymatrix} Y^{N-j}\\0 \end{mymatrix}\oplus \begin{mymatrix}0\\X^{N-j}\end{mymatrix}\right) \otimes\CC[F_1,F_2]\,.
 \] 

To obtain $\DD_N$-invariants we apply the projection $\frac{1}{2}(1+s)$.
Observe that the $\DD_N$-invariant polynomials move through this operator so that one only needs to compute 
\[\frac{1}{2}(1+s)\begin{mymatrix} X^{j}\\0 \end{mymatrix}
=\frac{1}{2}\begin{mymatrix} X^{j}\\Y^{j} \end{mymatrix}
=\frac{1}{2}(1+s)\begin{mymatrix} 0\\Y^{j} \end{mymatrix}\]
and 
\[\frac{1}{2}(1+s)\begin{mymatrix} Y^{N-j}\\0 \end{mymatrix}
=\frac{1}{2}\begin{mymatrix} Y^{N-j}\\X^{N-j} \end{mymatrix}
=\frac{1}{2}(1+s)\begin{mymatrix} 0\\X^{N-j} \end{mymatrix}
\]
These two vectors are independent over the ring by the previous reasoning. 

\end{proof}

\textbf{Remark.}
If one allows $\sigma$ to be non-faithful, several more cases appear. However, they are not more interesting than what we have seen so far, which is why we decided not to include this in the theorem. In words, it is as follows.
If the character of $\sigma$ is $\psi_{j'}$ and $\bigslant{\DD_N}{\ker \tau}$ is a subgroup of $\bigslant{\DD_N}{\ker \sigma}$ then everything is the same as above except that $N$ will be replaced by $N'=\frac{N}{\textrm{gcd}(N,j')}$. If on the other hand $\bigslant{\DD_N}{\ker \tau}$ is not a subgroup of $\bigslant{\DD_N}{\ker \sigma}$, then the only invariant vector is the zero vector.


\section*{References}
\def\cprime{$'$}

\end{document}